\documentclass[11pt,a4]{article}
\usepackage{color}
\usepackage{epsfig}
\usepackage{latexsym}
\usepackage{graphicx} 
\usepackage{subfigure}
\usepackage{amssymb}
\usepackage{url}     
  
\newtheorem{theorem}{Theorem}
\newtheorem{cor}[theorem]{Corollary}

\newtheorem{lemma}[theorem]{Lemma}

\newtheorem{invariant}[theorem]{Invariant}

\newcommand{\cm}[1]{}

\newcommand{\head}{\ensuremath{\mathrm{head}}}
\newcommand{\tail}{\ensuremath{\mathrm{tail}}}

\newenvironment{proof}{\noindent{\bf Proof:}}{
\hspace*{\fill} $\Box$ \vskip \belowdisplayskip}

\begin{document}
\title{Egalitarian Graph Orientations} \author{Glencora
  Borradaile\footnote{Supported by NSF CCF-0963921 and NSF
    0852030.}\\Oregon State University \and Jennifer
  Iglesias\footnote{Work done while at Oregon State University, Math
    REU.  Supported by NSF 0852030.}\\Carnegie Mellon University \and Theresa
  Migler\footnote{tmigler@calpoly.edu} \footnotemark[1]\\Oregon
  State University \and Antonio Ochoa\footnotemark[2]\\Cal Poly Pomona
  \and Gordon Wilfong\\Bell Labs \and Lisa Zhang\\Bell Labs}

\maketitle

\begin{abstract}
 Given an undirected graph, one can assign directions to each of the
 edges of the graph, thus {\em orienting} the graph.  To be as
 egalitarian as possible, one may wish to find an orientation such
 that no vertex is unfairly hit with too many arcs directed into it.
 We discuss how this objective arises in problems resulting from
 telecommunications.  We give optimal, polynomial-time algorithms
 for: finding an orientation that minimizes the lexicographic order
 of the indegrees and finding a strongly-connected orientation that
 minimizes the maximum indegree. We show that minimizing
 the lexicographic order of the indegrees is NP-hard when the
 resulting orientation is required to be acyclic.

 {\em keywords: algorithms, graph orientation, routing algorithms}
\end{abstract}

\section{Introduction}
\label{s:intro}
We consider problems of orienting the edges of an undirected graph so
that no vertex is unfairly hit with too many arcs directed into it.
We refer to such orientations as {\em egalitarian}: the total
available indegree is shared among the vertices as equally as allowed by the topology of
the graph. This objective arises in various telecommunications
problems. Depending on the requirements of the problem, the
orientation may be unconstrained or need to be strongly connected or
acyclic. We start by describing these motivating applications and
related work.

\paragraph{Unconstrained orientations}
Venkateswaran introduced the problem of directing the edges of an
undirected graph so as to minimize the maximum
indegree~\cite{venkateswaran2004}.  The problem arises from a
telecommunications network design problem in which source-sink pairs
$(s_i,t_i)$ are linked by a directed $s_i$-to-$t_i$ path $c_i$ (called a {\em circuit}).  When
an edge of the network fails, all circuits using that edge fail
and must be rerouted.  For each failed circuit, the responsibility for
finding an alternate path is assigned to either the source or sink
corresponding to that circuit. To limit the rerouting load of any vertex,
it is desirable to minimize the maximum number of circuits for which
any vertex is responsible.  

Venkateswaran  models this problem with an undirected
graph whose vertices are the sources and sinks and whose edges are the
circuits. He assigns the responsibility of
a circuit's potential failure by orienting the edge to either the
source or the sink of this circuit.  Minimizing the maximum number of
circuits for which any vertex is responsible can thus be achieved by finding an orientation that minimizes the maximum
indegree of any vertex.  Venkateswaran shows how to find such an
orientation~\cite{venkateswaran2004}. Asahiro, Miyano, Ono, and Zenmyo give a simpler
analysis~\cite{amoz2006}. Asahiro et al show further that for
any $w\geq \lceil {{{max degree}}\over {2}}\rceil$ the {\sc Path
  Reversal} algorithm minimizes the number of vertices with indegree
at most $w$ and consequently minimizes the number of vertices with
indegree at least $w+1$~\cite{ajmo2016}.

However, there may be multiple orientations that have the same
minimized maximum indegree.  The orientation that has the minimum
number of nodes with maximum indegree is preferable since it minimizes
the number of nodes that have the maximum rerouting load.  Among the
orientations that minimize the number of nodes with maximum indegree,
the one that minimizes the number of second largest indegree is
preferable for the same reason of rerouting load.  Continuing this
reasoning, we can formalize this notion in the following way: given two
orientations $G_A$ and $G_B$, we prefer $G_A$ to $G_B$ if the sequence
of indegrees of $G_A$ (in non-increasing order) is lexicographically
before the sequence of indegrees of $G_B$ (in non-increasing order).
We refer to finding the best orientation with respect to this measure
as the {\em minimum lexicographic} orientation.  In
Section~\ref{sec:min-lex-arb}, we show that a natural greedy algorithm
finds the minimum lexicographic orientation: start with an arbitrary
orientation and repeatedly reverse the orientation of a directed path
while doing so improves the objective.

\paragraph{Strongly-connected orientations} Networks that are used to
route messages should, naturally, be strongly connected: one should be
able to send a packet along a directed path from any vertex to any
other vertex.  A destination-based routing protocol chooses the next
arc along which to send a message based only on the destination of the
message.  Such a protocol can be implemented with an {\em interval
  routing scheme}~\cite{intervalrouting}. An interval routing scheme
for a directed graph is defined by a cyclic numbering of the vertices
and a labeling of each arc with an interval of the vertex
numbers. (More generally, each edge can be labelled with multiple
intervals. We will show that one interval is sufficient and therefore
the best possible.)  For each vertex $u$, the disjoint union of the
intervals labeling the outgoing arcs from $u$ cover all but $u$'s
vertex numbers.  When a packet destined for vertex $v$ reaches a
vertex $u\not = v$ it is forwarded from $u$ along the outgoing arc
from $u$ whose label contains the interval containing $v$'s number.
Such a scheme, in order to be feasible, must be such that a packet
originating at any vertex destined for any other vertex will reach the
destination vertex when routing is done as described above.

In Section~\ref{sec:route-to-min-max}, we show that for any 
strongly connected graph there is an interval routing scheme such that
each outgoing arc is labelled with at most one interval.
This is the most compact routing scheme possible and allows the
routing decision at a given vertex to be made in time proportional to
the outdegree of that vertex.  Thus, to minimize
the routing time at each vertex, we would like to
find a strongly-connected orientation of $G$, the underlying physical network, which minimizes the maximum number of outgoing arcs from any
vertex.  To keep the notation the same between sections of this paper,
we instead minimize the maximum indegree; this is equivalent by way of
reversing all the edges of the graph.  We give an algorithm to find
such an orientation in Section~\ref{sec:min-max-sc}.  We conjecture
that the natural generalization of this algorithm also finds the
minimum lexicographic order of the indegrees of the graph.

\paragraph{Acyclic orientations}
Consider a packet network with input buffers.  A vertex can forward a
packet from its input buffer to the next-hop (the next vertex in the
packet's route) if the input buffer of the next-hop is not already
full.  Such networks can suffer from deadlock.  For example, consider
a ring network in which all input buffers are full: no vertex can
forward a packet to its next-hop because the next-hop's input buffer
is full.  If no packet is allowed to go along certain length-two paths
then deadlock is prevented.  In particular, Wittorff shows how to find
such a collection of forbidden length-two paths by orienting the edges
of the network so that the resulting graph is acyclic with a single
source and making every pair of edges oriented into the same vertex a
forbidden length-two path~\cite{Wittorff:2009:biblatex}.  Then a path
between every pair of vertices avoiding forbidden paths can be found
that avoids any transition from travelling along an arc to travelling
along the reverse of another arc (and hence avoids a pair of edges
that get directed into the same vertex). Minimizing the maximum
indegree minimizes the number of forbidden pairs at any vertex and
hence minimizes the number of routing contraints at any vertex.

In Section~\ref{app:acyclic} we present a simple algorithm to find an
acyclic orientation for the objective of minimizing the maximum
indegree. On the other hand, we also show that minimizing the lexicographic
order of the indegrees is NP-hard when the resulting orientation must
be acyclic.

\subsection{Related work}
\label{s:related}
Asahiro et al.~consider the edge-weighted version of the unconstrained
problem~\cite{amoz2006}.  They build on the work of Venkateswaran and
give a $2-1/k$-approximation algorithm where
$k$ is the maximum weight of any edge in the
graph.  They further show that the weighted version of the
problem is strongly NP-hard even if all edge weights belong to the set
$\{1,k\}$
where $k\geq
2$ is an integer~\cite{ajmoz2011}. Klostermeyer considers the problem of reorienting
edges (rather than whole paths) so as to create graphs with given
properties, such as strongly connected graphs and acyclic graphs
\cite{klostermeyer99}. De Fraysseix and de Mendez show that they can
find an indegree assignment of the vertices given a particular
properties \cite{fm1994}. In our work we are searching for a
particular degree assignment not known a priori.

Biedl, Chan, Ganjali, Hajiaghayi, and Wood give a $13\over
8$-approximation algorithm for finding an ordering of the vertices
such that for each vertex $v$, the neighbors of $v$ are as evenly
distributed to the right and left of $v$ as possible
\cite{bcghtw2005}.  For the purpose of deadlock
prevention~\cite{wimmer78}, Wittorff describes a heuristic for finding
an acyclic orientation that minimizes the sum over all vertices of the
function $\delta (v)$ choose $2$, where $\delta (v)$ is the indegree
of vertex $v$. This obective function is motivated by a
  problem concerned with resolving deadlocks in communications
  networks as described in the previous section~\cite{Wittorff:2009:biblatex}.

\subsection{Notation}

We use basic notation for graph theoretic concepts for graphs $G =
(V,E)$ with $n$ vertices and $m$ edges.  A directed edge or arc, $a$,
is oriented from the vertex $\tail(a)$ to the vertex $\head(a)$.  For
a directed graph, the {\em indegree} of a vertex $v$, denoted $\delta
(v)$, is the number of arcs for which $v$ is the head.  We may use a
subscript to denote the graph with respect to which we measure the
degree.  A {\em directed path} is a sequence of arcs $a_1,a_2,\ldots,a_k$
with $\head(a_i) = \tail(a_{i+1})$ for $1\le i <k$. We add trivial
paths to this definition which are identified by a single vertex. A cycle is a
path such that $\head(a_k) = \tail(a_1)$.  An {\em orientation} of an
undirected graph is an assignment of directions to each edge in the
graph.  A directed graph is {\em strongly connected} if for every pair
of vertices, $u,v\in V$, there are directed paths from $u$ to $v$ and
from $v$ to $u$. A directed graph is {\em acyclic} if there are no
directed cycles in the graph.  For a subset of vertices $X$, $G[X]$ is
the subgraph induced by $X$ and $m(X)$ is the number of edges in
$G[X]$.

\section{Unconstrained orientations}
\label{sec:unconstrained}

We will show that a simple, greedy algorithm, first given by de
Fraysseix and de Mendez~\cite{fm1994}, finds an orientation of an
undirected graph that minimizes the lexicographic order of the
indegrees.  We say that a directed path from $u$ to $v$ is
{\em{reversible}} if $\delta (u) < \delta (v) - 1$.  The greedy
algorithm, given an undirected graph, is:
\begin{tabbing}
  {\sc Path-Reversal}\\
  \qquad \= arbitrarily orient every edge\\
  \qquad \= while there is a reversible path \\
  \> \qquad \= let $P$ be any reversible path whose last vertex
  is of highest indegree\\
  \> \> reverse the orientation of each arc of $P$
\end{tabbing}

This algorithm can be implemented in quadratic time by arguing that
the algorithm proceeds in $k$ phases where $k$ is the maximum indegree
in the initial orientation (below). Therefore, there are at most $m$
iterations of the algorithm and each iteration can be implemented in
linear time using, for example, depth-first search.

Consider any integer $\ell \le k$.  Consider an iteration in which we
reverse a $u$-to-$v$ path where $\delta(v) = \ell$.  Let $Q$ be the
set of vertices of indegree $> \ell$ just before this reversal and let
$Q'$ be the set of vertices that have paths to a vertex in $Q$.
(Note: $Q \subseteq Q'$.)  By definition neither $u$ nor $v$ is in
$Q'$, for otherwise, we would reverse a path ending in a vertex of
indegree $> \ell$.  Further, after this reversal, $Q$ is still the set
of vertices of indegree $> \ell$ and $Q'$ is still the set of vertices
that have paths to a vertex in $Q$.  It follows that there is a
well-defined phase $\ell$, a contiguous subset of iterations that
reverse paths ending in vertices of indegree $\ell$: after reversing a
path ending in a vertex of indegree $\ell$, the algorithm does not
reverse a path ending in a vertex of higher indegree.

\subsection{Minimizing the lexicographic order} \label{sec:min-lex-arb}
{\sc Path-Reversal} finds an orientation that
minimizes the maximum indegree. This observation was made by
Venkateswaran with a rather involved proof~\cite{venkateswaran2004}; a
simpler analysis was given by Ashario et~al.~\cite{amoz2006}.  This
observation is also implied by de Fraysseix and de Mendez,
Lemma~1~\cite{fm1994}.

{\sc Path-Reversal} is more powerful than simply minimizing the
maximum indegree.  We show that the resulting orientation, in fact,
minimizes the lexicographic order of the indegrees.  

We define a {\em cycle reversal} to be the reversal of every edge in a
cycle. Notice that performing a cycle reversal will not change the
number of vertices of any particular indegree.

\begin{lemma}
\label{cycle_reversal}
  Let $O_1$ and $O_2$ be orientations such that $\delta _{O_1} (v) =
  \delta _{O_2} (v)$ for all $v$. Then $O_1$ can be transformed into
  $O_2$ by a sequence of cycle reversals. 
\end{lemma}

\begin{proof}
  Let $E_R$ be the set of arcs in $O_1$ that have an opposite
  orientation in $O_2$. Notice that for all vertices $v$ in the graph
  induced by $E_R$, the indegree of $v$ equals the outdegree of
  $v$. For otherwise, there is some $v$ such that $\delta _{O_1} (v)
  \neq \delta _{O_2} (v)$. It follows that each connected component of
  the graph induced by these
  edges has an Euler tour~\cite{gy2004}. Reversing these tours gives the
  lemma.
\end{proof}

We define a {\em weak reversal} to be the reversal of a path from a
vertex $u$ to a vertex $v$ where $\delta(u) =\delta(v)-1$. Notice that performing a weak reversal will not change the number of vertices
of any particular indegree. 

\begin{theorem}
  Any orientation that minimizes the lexicographic order of the
  indegrees of the vertices can be transformed into an orientation
  induced by {\sc Path-Reversal} via a sequence of weak reversals or cycle reversals.
\end{theorem}
\begin{proof}
  Let $D_{lex}$ denote an orientation that minimizes the lexicographic
  order of the indegrees of the vertices, and let $D_{PR}$ denote an
  orientation given by {\sc Path-Reversal}. Let $\delta _{lex} (v)$
  and $\delta _{PR} (v)$ be the indegree of a vertex $v$ in $D_{lex}$
  and $D_{PR}$ respectively. 

We will use induction on $S := \sum _
  {v\in V} |\delta _{lex} (v) -\delta _{PR} (v)|$. If $|S|=0$ then by
  Lemma~\ref{cycle_reversal}, the theorem holds.
  Now suppose that $S>0$. Let $S_{\neq} = \{v:\delta _{lex} (v) \neq
  \delta _{PR} (v)\}$. Let $v$ be a vertex in $S_{\neq}$ that 
  maximizes $\delta _{lex} (v)$ and if there is a choice among many
  such vertices, then maximizes $\delta _{PR} (v)$.
  Then we have the following two cases:
  \begin{enumerate}
  \item $\delta _{lex} (v) > \delta _{PR} (v)$. Let $U$ be the set of
    all vertices that can reach $v$ in $D_{lex}$. Notice that 

    \begin{equation}
      \label{eq:1}
    \sum _{u\in U} \delta _{lex} (u) \leq \sum _{u\in U} \delta _{PR}
    (u)
    \end{equation}

    This is because $\sum _{u\in U} \delta _{lex} (u)$ is the number
    of edges in $G[U]$, $\sum _{u\in U} \delta _{PR} (u)$ also
    includes the indegree from edges in $U$ and may additionally
    include the indegree from edges directed into $U$.
    $\delta _{lex} (v) > \delta _{PR} (v)$ and $v\in U$, so there must
    be some $u\in U$ with $\delta _{lex} (u) < \delta _{PR}
    (u)$.
    Because we chose $v$ to maximize $\delta _{lex} (v)$,
    $\delta _{lex} (u) \leq \delta _{lex} (v)$. Furthermore,
    $\delta _{lex} (u) \neq \delta _{lex} (v)$ because if
    $\delta _{lex} (u) = \delta _{lex} (v)$ then
    $\delta _{PR} (u) > \delta _{lex} (u) = \delta _{lex} (v) >
    \delta_{PR} (v)$,
    but we chose $v$ to maximize $\delta _{PR} (v)$. Therefore
    $\delta _{lex} (u) < \delta _{lex} (v)$. It is not possible for
    $\delta _{lex} (u) < \delta _{lex} (v)-1$ for otherwise reversing
    a $u$ to $v$ path would give an orientation with a smaller
    lexicographic order than $D_{lex}$. Therefore we have that
    $\delta _{lex} (u) = \delta _{lex} (v)-1$ and there is a weakly
    reversible path from $u$ to $v$ in the directed graph defined by
    $D_{lex}$. Reversing this path decreases $S$ by 2.

\item $\delta _{lex} (v) < \delta _{PR} (v)$. Let $U$ be the set of
  vertices that reach $v$ in $D_{PR}$. Notice that $\sum
    _{u\in U} \delta _{lex} (u) \geq \sum _{u\in U} \delta _{PR}
    (u)$. $\delta _{lex} (v) < \delta _{PR} (v)$ and $v\in U$, so
    there must be some $u\in U$ with $\delta _{lex} (u) > \delta _{PR}
    (u)$. We chose $v$ to be
  a vertex that maximizes $\delta _{lex} (v)$, so $\delta _{lex} (u)
  \leq \delta _{lex} (v)$. So we have $\delta _{PR} (u) < \delta
  _{lex} (u) \leq \delta _{lex} (v) < \delta _{PR} (v)$. This means
  that there is a reversible path from $u$ to $v$, a contradiction
  because $D_{PR}$ has no reversible paths.
  \end{enumerate}
\end{proof}

\begin{cor}
\label{opt_lex}
The algorithm {\sc Path-Reversal} finds an orientation that minimizes the
lexicographic order of the indegrees.
\end{cor}

\paragraph{Remarks}  
Let us revisit the motivating problem of failure recovery in network
design, for which a failed circuit notifies either its source vertex
or its sink vertex.  We argued that minimizing the maximum indegree
or the lexicographic order of the indegrees
 minimizes effort in failure recovery.  We could also measure 
the recovery effort  per vertex as
a function $f(\cdot)$ of the number of circuits this vertex is
responsible for. 
The total effort for error recovery is then $\sum_v f(\delta(v))$. 
The shape of $f(\cdot)$, convex or concave or other more complex nature,
can be debated. However, if $f(\cdot)$ is increasing
and strictly convex, we remark that the algorithm {\sc Path-Reversal}
as we have seen also minimizes the total effort. 
Asahiro et al.~\cite{ajmo2012} present a network flow algorithm that
also gives an orientation that minimizes $\sum_v f(\delta^+ (v))$
where $\delta^+ (v)$ denotes the outdegree of a vertex $v$ when $f$ is convex.

\begin{theorem}
\label{opt_convex}
The algorithm {\sc Path-Reversal} finds an orientation $G$ that minimizes
$F(G) = \sum_v f(\delta_G(v))$ for any increasing and strictly convex
function $f$.
\end{theorem}

\begin{proof}
  Let $\alpha_i(G)$ denote the number of vertices of indegree $i$ in
  $G$. We rewrite the objective to be $F(G) = \sum_i \alpha_i(G)\cdot f(i)$.
  
  Let $G_a$ be an orientation of the graph that minimizes the given
  objective.  Let $G_b$ be the result of the algorithm {\sc Path-Reversal}
  using $G_a$ as the initial orientation.  By Theorem~\ref{opt_lex},
  $G_b$ minimizes the lexicographic order of the indegrees.  Since the
  non-increasing sequence of indegrees that corresponds to such an
  orientation is unique, it follows that any orientation $G_c$ that
  minimizes the lexicographic order of the indegrees satisfies
  $\alpha_i(G_c) = \alpha_i(G_b)$ for all $i$.  It further follows
  that all orientations that minimize the lexicographic order of the
  indegrees achieve the same objective: $F(G_c) = F(G_b)$.

  Suppose for a contradiction to the theorem that the degree
  distributions of $G_b$ and $G_a$ differ.  Therefore, the algorithm
  {\sc Path-Reversal} performs at least one path reversal.  Let $G_1$
  be the graph obtained from $G_a$ after reversing one path, say from
  a vertex $u$ to a vertex $v$. 

  We compare $F(G_a)$ to $F(G_1)$.  Let $\delta_{G_a}(u) = k$
  and $\delta_{G_a}(v) = \ell$. Since this path from $u$ to $v$ was a
  reversible path in $G_a$, $k
  < \ell - 1$.  By the path reversal operation, we get
  $\delta_{G_1}(u) = k+1$ and $\delta_{G_1}(v) = \ell-1$.
  Then:
  \begin{eqnarray*}
    F(G_a)-F(G_1) &=& 
    \sum_x f(\delta_{G_a}(x)) -\sum_x f(\delta_{G_1}(x)) \\
    &=& f(\delta_{G_a}(v))-f(\delta_{G_1}(v))+f(\delta_{G_a}(u))-f(\delta_{G_1}(u)) \\
    &=& \underbrace{f(\ell)-f(\ell-1)}_{A}-\underbrace{\left(f(k+1)-f(k)\right)}_{B} \\
  \end{eqnarray*}
  Since $k < \ell - 1$ and $f$ is increasing and strictly convex, term $A$ is
  greater than term $B$, and so the above difference is positive.  It
  follows that $F(G_a) > F(G_1)$, contradicting the fact that $G_a$ minimizes the objective $F(G)$. Therefore, $G_a$ and $G_1$ must have the same degree distribution.
\end{proof}

\section{Strongly connected orientations} \label{sec:route-to-min-max}

In this section we will show how to find a strongly-connected
orientation that minimizes the maximum indegree. First we argue that
this would enable an interval routing scheme (as described in the
introduction) with minimum table sizes. A routing table for a vertex
$v$ assigns intervals to each outgoing arc that encode how a message
should leave $v$. The size of a table for a given vertex $v$ is the number of intervals
summed over all outgoing arcs from $v$.

\subsection{Minimum routing tables for strongly connected graphs}

It is well known, as a generalization of Whitney's characterizations of 2-edge connected, undirected graphs~\cite{Whitney32} and Robbins' correspondence between strong connectivity and 2-edge connectivity~\cite{Robbins39}, that a directed graph is strongly connected if and only if it has an ear-decomposition.  An {\em ear decomposition} of a directed graph is a partition of the edges into a simple directed cycle $P_0$ and simple directed paths (or cycles) $P_1, \ldots, P_k$ such that for each $i > 0$, the intersection of $P_i$ with $\cup_{j < i} P_j$ are the endpoints of $P_i$ (which may be coincident if $P_i$ is a cycle).  Each $P_i$ is called an ear.

An ear decomposition can be found in linear time~\cite{et1976}.  Given an ear decomposition of a strongly-connected graph, we can define the routing tables using the procedure {\sc Routing} below.  We will define a cyclic ordering $\cal L$ of the vertices.  For each arc $uv$, we will define an interval $[a,b]$, $a,b \in \cal L$.  Recall from the introduction that this information can be used for routing: a message at a vertex $u$ with destination $d$ will be forwarded along $uv$ if $d$ is in $[a,b]$, that is if $d$ is (inclusively) between $a$ and $b$ in the cyclic ordering $\cal L$.  We say that such a labeling is feasible if it allows a message to be routed between any pair of vertices.

We assume, without loss of generality, that each ear in the ear decomposition contains at least two edges: a single-edge ear could be removed while maintaining strong connectivity and so will not be required for routing. We denote the number of edges in $P$ by $|P|$. See Figure~\ref{fig:routing} for a demonstration of this procedure.  It is convenient to use both open and closed endpoints for intervals of $\cal L$.  For example, $(a,b]$ contains all the vertices that are strictly after $a$ and before (or equal to) $b$ in the ordering.  We use $(a,a)$ to represent all the vertices in the cyclic ordering except $a$.  Further, for the purposes of analysis, it is convenient to think of the intervals as being continuous.

\begin{tabbing}
  {\sc Routing} (ear decomposition $P_0 ,P_1,\dots$)\\
  \qquad \= Initialize $\cal L$ to contain all the vertices of $P_0$ in their order around $P_0$.\\
  \> Assign each arc $ab$ of $P_0$ the interval $(a,a)$. \\
  \> For $i = 1, \ldots, k$: \\
  \> \qquad \= Let $v_1$ be the first vertex of $P_i$.\\
  \> \> Let $v_2, \ldots, v_p$ be the second through penultimate vertices of $P_i$. \\
  \> \> Insert $v_2, \ldots, v_p$ into $\cal L$ after $v_1$. \\
  \> \> For $j = 2, \ldots, p$: \\
  \> \> \qquad \= Assign the arc leaving $v_j$ the interval $(v_j,v_j)$.\\
  \> \> Let $v_1u$ be the arc leaving $v_1$ that is assigned the interval $(v_1,a)$ (for some $a$).\\
  \> \> Let $b$ be the vertex after $v_p$ in the cyclic ordering $\cal L$.\\
  \> \> Reassign $v_1u$ the interval $[b,a)$. \\
  \> \> Assign the arc $v_1v_2$ the interval $(v_1,b)$.
\end{tabbing}

The following invariant, among other things, shows that the arc $v_1u$ exists and is unique.
\begin{invariant}
  At any stage in the algorithm, the intervals assigned to the arcs leaving a vertex $v$ are disjoint and form a partition of $(v,v)$.
\end{invariant}

\begin{proof}
  When a vertex is first introduced and there is only one arc leaving it, this invariant is true by construction.

  If we assume for an induction that the invariant holds prior to the introduction of a new arc $v_1v_2$ leaving $v_1$, then there must be exactly one arc whose assigned interval starts with $(v_1,$.  (Also, since a closed endpoint of an interval is never introduced, this arc must be assigned an interval of the form $(v_1,a)$ for some $a$.)  Since, prior to the insertion of $v_2, \ldots, v_p$  into $\cal L$, $b$ is the vertex after $v_1$ in the cycle ordering, $(v_1,b) \subseteq (v_1,a)$ for all $a$.  Therefore after the insertion of $v_2, \ldots, v_p$  into $\cal L$ between $v_1$ and $b$ we still have that $(v_1,b) \subseteq (v_1,a)$.  Since $(v_1,b),[b,a)$ is a partition of $(v_1,a)$, the invariant holds.
\end{proof}

\begin{theorem}
  {\sc Routing} produces a feasible interval routing scheme with each
  arc having exactly one interval.
\end{theorem}

Note that this result has been shown previously with a different
approach by Fraigniaud and Gavoille~\cite{fg1994} in Lemma 3.

\begin{proof}
  As mentioned above, for convenience of analysis, we view the intervals as continuous.  Let $H_i = \cup_{j \le i} P_j$ and let ${\cal L}_i$ be the cyclic ordering of the vertices of $H_i$ at the start of iteration $i$ (or end of iteration $i-1$ for $i = k$).  For a vertex $v \in H_i$, let $n_i(v)$ be the vertex immediately after $v$ in ${\cal L}_i$.  We prove the following statement by induction: for $v \in H_i$, at the start of iteration $i$ (or end of iteration $i-1$ for $i = k$), a message with destination in the (continuous) interval $[v, n_i(v))$ will reach vertex $v$.  This statement is true for the base case which corresponds to the interval assignment for $P_0$.

Consider ear $P_i$.  We show that the intervals defined at the end of iteration $i$ allow a message with destination in the interval $[y, n_i(y))$ starting at vertex $x$ will reach vertex $y$ (for $x \ne y$).  The non-trivial cases are Cases~\ref{xv1yinPi},~\ref{xHyP} and~\ref{xv1yH}.
  \begin{enumerate}
  \item $x$ and $y$ are internal vertices of $P_i$ and $y$ is after $x$ along $P_i$:\\\label{xinPiyinPi}
    For every vertex $v$ in $P_i$, a message will get routed on the arc leaving $v$ unless it is destined for $v$ since the interval assigned to the unique arc leaving $v$ contains everything except $v$.
  \item $x = v_1$ and $y$ is an internal vertex of $P_i$:\\\label{xv1yinPi}
    The vertices in $P_i$ are in the interval $[v_2,v_p]$.  By construction and definition of $b$, this is the same as the interval $(v_1,b)$ since $b$ is the vertex after $v_p$ in $\cal L$ and $v_1$ is the vertex before $v_2$ in $\cal L$.   So $[y, n_i(y)) \subset (v_1,b)$ and a message at $v_1$ going to a destination in $[y, n_i(y))$ gets routed along the arc $v_1v_2$.  Correctness follows from Case~\ref{xinPiyinPi}.
  \item $x \in H_{i-1}\setminus v_1$ and $y$ and internal vertex of $P_i$:\\\label{xHyP}
    We argue that the message will reach $v_1$.  By definition $b = n_{i-1}(v_1)$ and by construction $[y,n_i(y)) \subset [v_1, b) = [v_1, n_{i-1}(v_1))$.  Since $v_1 \in H_{i-1}$, by the inductive hypothesis, a message with destination in the interval $[v_1, n_{i-1}(v_1))$ will reach $v_1$; we are done by Case~\ref{xv1yinPi}.
  \item $x = v_1$ and $y \in H_{i-1}$:\\\label{xv1yH} 
    Note that $n_{i-1}(y) = n_i(y)$.  Since $y \ne x$ and $b = n_{i-1}(x)$, $y \in [b,v_1)$.  Therefore $[y,n_i(y)) \cap (v_1,b)$ is empty and the message does not get routed along $v_1v_2$.  Therefore a message in $[y,n_i(y))$ reaches $y$ by the inductive hypothesis.
  \item $x,y \in H_{i-1}$:\\\label{xyH}
    If a message in $[y,n_i(y))$ reaches $v_1$, then the message reaches $y$ by Case~\ref{xv1yH}.  If a message in $[y,n_i(y))$ does not reach $v_1$, then we are done by the inductive hypothesis because $n_{i-1}(y) = n_i(y)$.
    \item $x$ is an internal vertex of $P_i$ and $y \in H_{i-1}$:\\
      Note that $n_{i-1}(y) = n_i(y)$.  Since $v \notin [y,n_i(y))$ for any internal vertex $v$ of $P_i$, a message in $[y,n_i(y))$ will reach $H_{i-1}$.  Then by Case~\ref{xyH}, a message in $[y,n_i(y))$ will reach $y$.
    \item $x$ and $y$ are internal vertices of $P_i$ and $y$ is before $x$ along $P_i$:\\
      Since $v \notin [y,n_i(y))$ for any internal vertex $v$ of $P_i$ after $x$ because $x$ is after $y$ in $P_i$, a message in $[y,n_i(y))$ will reach $H_{i-1}$.  By Cases~\ref{xv1yinPi} and~\ref{xHyP}, a message in $[y,n_i(y))$ will reach $y$.
  \end{enumerate}
\end{proof}

It is non-standard to use open intervals for such a scheme.  Given the final interval assignment and cyclic ordering, numbers can be assigned to the vertices based on the cyclic ordering and the intervals can be closed in the natural way.

\begin{figure}[ht]
  \centering
  \includegraphics[scale=.35]{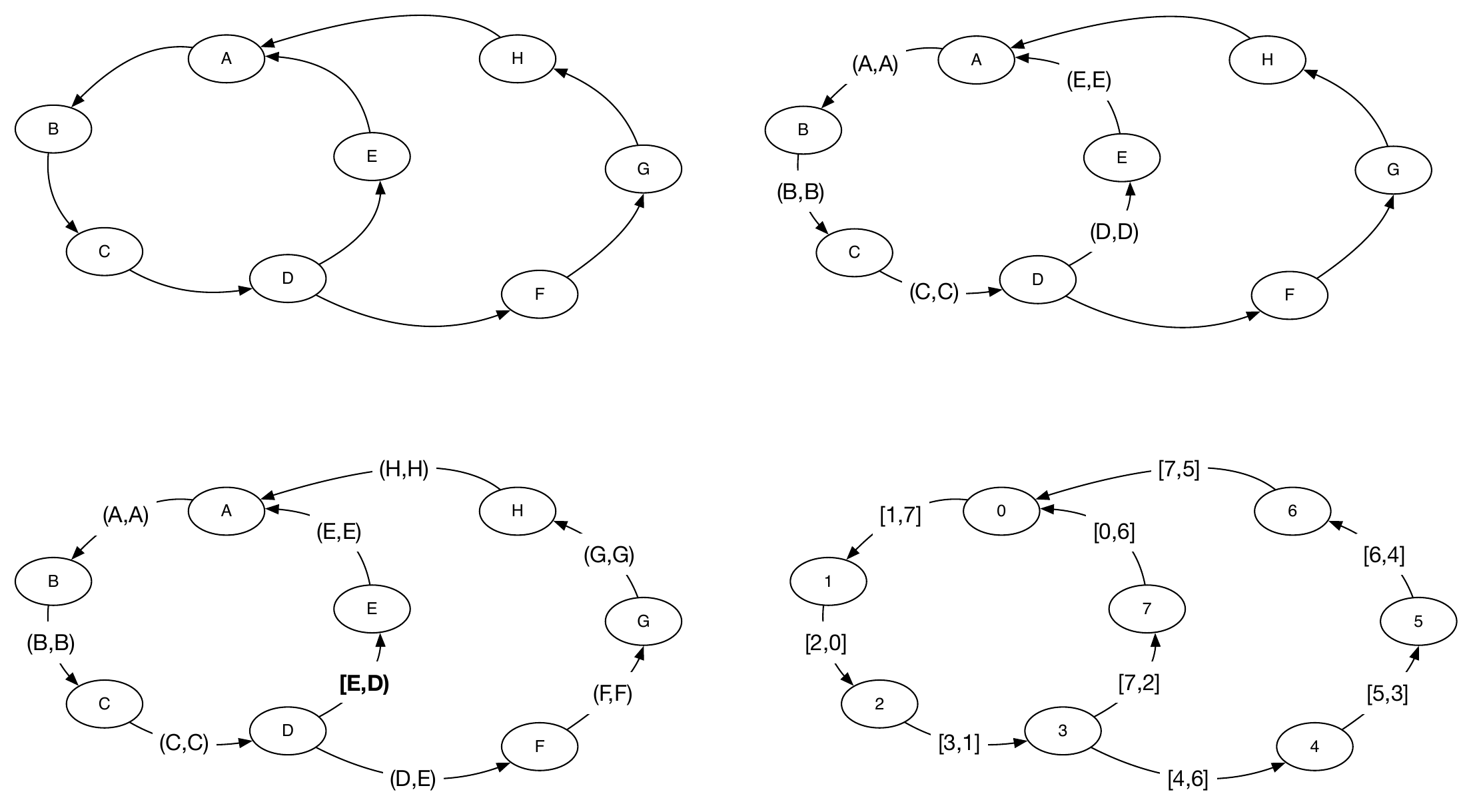}
  \caption{Top left: Input strongly connected component with symbolic node labels and two ears, $P_0 = A, B, C, D, E$ and $P_1 = D, F, G, H, A$.  Top right: Arc labeling of $P_0$ with cyclic ordering ${\cal L} = \{A,B,C,D,E\}$; $(C,C)$ indicates the range of $\cal L$ strictly after $C$ and strictly before $C$ (namely, $D,E,A,B$).  Bottom left: Arc labeling of $P_1$, update of the label for the arc(s) ($DE$) leaving the first node of $P_1$ ($D$), and inserting the internal nodes $F,G,H$ of $P_1$ into the cyclic ordering before the last node of $P_1$; 
${\cal L} = \{A,B,C,D, F, G, H, E\}$;  $[E,D)$ indicates the range of $\cal L$ after and including $E$ and strictly before $D$ (namely, $E, A, B, C$).  Bottom right: conversion to a numerical scheme with closed (cyclic) intervals by mapping the $i^{th}$ element of $\cal L$ to the number $i$.}
  \label{fig:routing}
\end{figure}

Since we can generate an interval routing scheme with exactly one
interval per arc and each arc is required for routing when each ear
has at least two arcs, our labeling is optimal. We can minimize the
table sizes if we can first strongly orient the graph to minimize the
maximum outdegree. To keep with the notation of the rest of the paper,
we instead, without loss of generality, minimize the maximum indegree.

\subsection{Strongly-connected orientations that minimize the maximum
  indegree} \label{sec:min-max-sc}

We will show that a modified version of {\sc Path-Reversal} finds a
strongly-connected orientation of an undirected graph that minimizes
the maximum indegree.  In this section we will assume that the given
directed graph has a strongly-connected orientation. Given a directed
graph, we say that a path from $u$ to $v$ is {\em strongly reversible}
if $\delta (u) < \delta (v) -1$ and reversing the path will maintain
strong connectivity.  The greedy algorithm, given an undirected graph,
is:
\begin{tabbing}
  {\sc SC-Path-Reversal}\\
  \qquad \= start with an arbitrary strongly-connected orientation\\
  \qquad \= while there is a strongly reversible path starting with a 
max-indegree vertex\\
  \> \qquad \= let $P$ be such a path \\
  \> \qquad \= reverse the orientation of each arc of $P$ 
\end{tabbing}

One can find a strongly-connected orientation in linear time using
depth-first search: orient all
edges in the depth-first search tree downward away from the root and orient all the
non-tree edges upward with respect to the tree, cross edges may be
oriented arbitrarily.

Strongly-reversible paths are characterized by the number of edge
disjoint paths between endpoints. We say that a vertex $v$ {\em
  two-reaches} a vertex $u$ if there are two arc-disjoint paths from
$v$ to $u$. We say that a vertex $v$ {\em two-reaches} a vertex set $U$ if
there are paths from $v$ to $u_1$ and from $v$ to $u_2$
where $u_1 ,u_2 \in U$, and these paths are arc disjoint.

In the following we will use network flow theory.  Let $N = (V, E)$ be
a directed network with $s, t \in V$ being the source and the sink of $N$
respectively. The capacity of an edge is a mapping $c : E \to
\mathbb{R}^+$, denoted by $ c_{uv}$. This is the maximum amount of
flow that can pass through an edge. A {\em flow} is a mapping $f : E
\to \mathbb{R}^+$, denoted by $f_{uv}$, subject to the following two
constraints: $f_{uv} \leq c_{uv}$, for each $(u, v) \in E$ and
$\sum_{u:(u, v) \in E} f_{uv} = \sum_{u:(v, u) \in E} f_{vu}$, for
each $v \in V \setminus \{s, t\}$. The value of flow is defined by
$|f| = \sum_{v \in V} f_{sv}$, where $s$ is the source of $N$. An
  $s-t$ {\em cut} $C = (S,T)$ is a partition of $V$ such that $s\in S$ and
$t\in T$. The cut-set of $C$ is the set $\{(u,v)\in E | u\in S, v\in
T\}$. The max-flow, min-cut theorem states that the value of the maximum
flow is equal to the value of the minimum cut~\cite{FF56}.

A consequence of this theorem is that if the maximum flow is greater
than or equal to $k$ in a unit-capacity network, then there are $k$
arc disjoint $s$ to $t$ paths \cite{amo1993}.

\begin{lemma} \label{lem:strong-connect-character} Reversing a
  $u$-to-$v$ path maintains strong connectivity if and only if $u$
  two-reaches $v$.
\end{lemma}

\begin{proof}
  Suppose that when we reverse a $u$-to-$v$ path $P$, the graph
  remains strongly connected. Thus there must still be a $u$-to-$v$
  path when $P$ is reversed, so in the original graph the max
  $u$-to-$v$ flow must have been at least 2. By the max-flow, min-cut
  theorem, we know that
  there are 2 arc-disjoint paths from $u$ to $v$.

  Now suppose that $u$ two-reaches $v$ in a strongly connected
  orientation. Reversing one of these paths will create a cycle. Any
  pair of vertices requiring one of these paths for connectivity can be
  connected by way of the cycle, which will maintain strong
  connectivity.
\end{proof}

{\sc SC-Path-Reversal} can be implemented to run in quadratic time:
strong-path reversibility can be detected in linear time by two
iterations of the augmenting path algorithm for maximum flow~\cite{FF56}.  There
are a linear number of iterations: we show, as in
Section~\ref{sec:unconstrained}, that after reversing a
strongly-reversible path ending in a vertex of indegree $\ell$, the
algorithm does not reverse a strongly-reversible path ending in a
vertex of higher indegree.

As for the algorithm {\sc Path-Reversal}, we argue that the algorithm
proceeds in $k$ phases where $k$ is the maximum indegree of the
initial strongly connected orientation. In phase $\ell$,
strongly-reversible paths ending in vertices of indegree $\ell$ are
reversed. This reduces the indegree of these vertices by one, and does
not result in any extra vertices of indegree greater than $\ell$.

Let $Q$ be the set of vertices of indegree $> \ell$ just before the
start of phase $\ell$ and let $Q'$ be the set of vertices that have
strongly-reversible paths to a vertex in $Q$.  (Note: $Q \subseteq
Q'$.)  If, in the first iteration of phase $\ell$, a
strongly-reversible path ending in a vertex $v$ of indegree $\ell$ is
reversed, then the indegrees of the vertices in $Q'$ must be $>
\ell$ and $v \notin Q$.  Therefore, after reversing a
strongly-reversible path ending in a vertex of indegree $\ell$, the
algorithm does not reverse a path ending in a vertex of higher
indegree.

\subsubsection{Strongly connected structure}

To prove that {\sc SC-Path-Reversal} minimizes the maximum indegree,
we will use a transitivity-like property of arc-disjointness:

\begin{lemma}
  \label{two_paths}
  Suppose vertices $s$ and $t$ each two-reach a vertex $v$.  If there
  are arc-disjoint $u$-to-$s$ and $u$-to-$t$ paths, then
  $u$ two-reaches $v$.
\end{lemma}

\begin{proof}
  We argue that the min $u$-$v$ cut is at least 2, proving the lemma by
  the max-flow-min-cut theorem.
  Consider any $u$-$v$ cut (viewed as a bipartition of the vertices),
  $(A,B)$.  If $s \in A$, then the min cut is at least 2 (because the
  min $sv$-cut is at least 2).  Likewise if $t \in A$.  If both $s$
  and $t$ are in $B$, then the min $u$-$v$ cut is at least 2, as
  witnessed by the arc-disjoint $u$-to-$s$ and $u$-to-$t$ paths.
\end{proof}

In order to ensure strong connectivity, we get:

\begin{cor}
\label{sc_simple_lower_bound}
Let $U$ be a set of vertices. For each component $C$ of $G[U]$, there
must be at least one arc from each component of $G[V\setminus C]$ to
$C$.  
\end{cor}

We can in fact meet the implied lower bound:
 
\begin{lemma}
  \label{one_edge_sc}
  Let $v$ be a vertex of maximum indegree obtained
  by the {\sc SC-Path-Reversal} algorithm.  Let $U$ be the set of vertices
  that two-reach $v$.  Each component of $G[V\setminus U]$ has {\em
    exactly one} arc to $U$ in the {\sc SC-Path-Reversal} orientation.
\end{lemma}

\begin{proof}
  Let $C$ be a component of $G[V\backslash U]$ and suppose for a
  contradiction that there are multiple arcs from $C$ to $U$.  Let
  $v_1,v_2,v_3,\dots , v_p$ be vertices in $C$ that are tails of
  these arcs.  Let $C_i$ be the set of vertices in $C$ that reach
  $v_i$. We will argue that all of the $C_i$s are in fact the same, so
  there is only one $v_i$ that is the tail of an arc from $C$ into $U$.

  Since the graph is strongly connected, every vertex in $C$ reaches $U$
  and so must reach some $v_i$.  If $x \in C_i \cap C_j$ for some $i
  \neq j$, then by Lemma~\ref{two_paths}, $x$ two-reaches $U$,
  contradicting the definition of $U$.  Therefore $C_1,\ldots,C_p$ is
  a partition of $C$.  However, since $C$ is connected, there must be
  an arc $uv$ from, say, $C_i$ to $C_j$.  In which case, $u \in C_j$,
  a contradiction by the above case.  Therefore, the partition cannot
  contain more than one set. So there is only one arc from $C$ to $U$.
\end{proof}

\subsubsection{Minimizing the maximum indegree}

We show that the algorithm minimizes the maximum indegree by meeting
the following lower bound.  For a set of vertices $U$, let $c(U)$ be
the number of components of $G[V\backslash U]$.

\begin{lemma}
\label{sc_lower_bound}
The maximum indegree of any strongly connected orientation is at least
\[\max_{U \subseteq V} \left\lceil {m(U) + c(U) \over |U|} \right\rceil.\]
\end{lemma}

\begin{proof}
  The total indegree that must be shared amongst $U$ is at least the
  number of edges in $G[U] + c(U)$, where the second term follows from
  Corollary~\ref{sc_simple_lower_bound}.  By an averaging argument at
  least one vertex must have indegree at least $\left\lceil {m(U) +
      c(U) \over |U|} \right\rceil$.
\end{proof}

\begin{theorem}
  \label{sc_min_max}
  The algorithm {\sc SC-Path-Reversal} finds a strongly connected
  orientation that minimizes the maximum indegree.
\end{theorem}

\begin{proof}
  Let $k$ be the maximum indegree resulting from {\sc SC-Path-Reversal}.
  Let $v$ be a vertex of indegree $k$.  Let $U$ be the set of vertices
  that two-reach $v$.  By the termination criteria of the algorithm,
  all vertices in $U$ have indegree $k$ or $k-1$.  By
  Lemma~\ref{one_edge_sc}, the total indegree shared amongst $U$ is
  $m(U)+c(U)$.  We have that $|U| k\geq m(U)+c(U) > |U|(k-1)$.
  Dividing by $|U|$ yields $k\geq \frac{m(U)+c(U)}{|U|} > (k-1)$.  
  By Lemma~\ref{sc_lower_bound}, this is the best possible.
\end{proof}

Theorem~\ref{sc_min_max} has previously been proven in a
non-constructive manner by Frank~\cite{frank1980, frank2011}.

We conjecture that {\sc SC-Path-Reversal} is indeed optimal for the
``minimizing the lexicographic order'' objective as
well. Unfortunately our proof technique from
Section~\ref{sec:unconstrained} for minimizing the lexicographic order
of an arbitrary orientation does not follow through. For example, we
would need to consider the set of vertices $U$ that have at least two
paths to a vertex of highest indegree, but there could be a vertex $x$
on a path from $u\in U$ to $v$ that is not in $U$. In this case
inequality (\ref{eq:1}) does not hold. For this, and other reasons, a
different technique will needed to obtain this result.

\section{Acyclic orientations} 
\label{app:acyclic}

We now examine the situation in which the resulting orientation needs
to be acyclic. Unlike what we have seen, minimizing the lexicographic
order is no longer polynomially solvable.  However, a simple algorithm
guarantees optimality of minimizing the maximum indegree.

\subsection{Minimizing the maximum indegree}
The following simple procedure minimizes the maximum indegree
for an acyclic orientation.

\noindent {\bf{Stripping Procedure}}
Choose a vertex with minimum degree. Orient all incident edges into that
vertex. {\em Remove} that vertex and its adjacent edges. Repeat.

\begin{theorem}
  Stripping finds an acyclic orientation with maximum indegree
  minimized.
\end{theorem}

\begin{proof}
  Let $k$ be the maximum indegree resulting from stripping and let $v$
  be a vertex with indegree $k$. Thus at some iteration, $v$ had the
  minimum degree among the remaining vertices, $U$. Let $H$ be the
  subgraph induced by $U$ with orientation inherited from an optimal
  orientation of the original graph. $H$ must have a sink and since
  every vertex in $H$ has degree at least $k$, this vertex must have
  indegree at least $k$.
\end{proof}

\subsection{Acyclic minimum lexicographic orientation is NP-hard}
We will show that the problem of minimizing the lexicographic order of
an acyclic orientation is NP-hard. We will give a reduction from set
cover to the related problem of finding an acyclic orientation such
that:
\begin{itemize}
\item the maximum indgree is minimized 
\item the number of nodes with this
  maximum indegree is also minimized
\end{itemize}
Clearly, finding an acyclic orientation with minimum lexicographic
order of indgrees solves this related problem. The decision version of
this related problem ``Is there an acyclic orientation that minimizes
the maximum inegree and further has at most $\ell$ vertices with that
maximum indegree?''  is in NP because the orientation is the
certificate.

The {\em Set Cover} problem is defined as follows: Given a set of
elements $\{1,2,\dots , m\}$ (called the universe) and $n$ sets whose
union comprises the universe, the set cover problem is to identify the
smallest number of sets whose union contains all elements in the
universe. Set cover is proven to be NP-hard by a reduction from the
vertex cover problem~\cite{gj1979}.

We say that a set of vertices is {\em $t$-strippable} if the stripping procedure
results in maximum indegree at most $t$ among these vertices.

We will use the following {\em $k$-gadget graph} $H_\ell$ with $1\leq \ell <k$ and $k$ odd:

  \begin{figure}[ht]
    \centering

    \subfigure[$k=5$, $\ell= 3$]{
    \includegraphics[scale=.57]{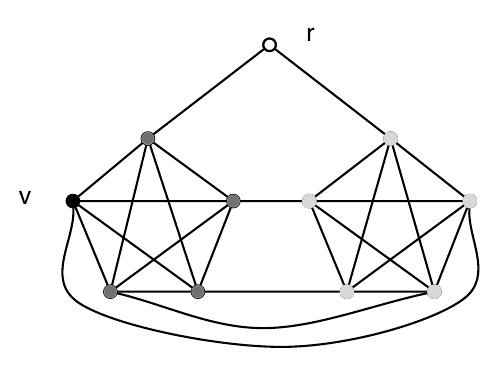}
    \label{fig:gadgets1}
  } 
  \subfigure[$k=5$, $\ell =2$ and $v\neq s$]{
    \includegraphics[scale=.57]{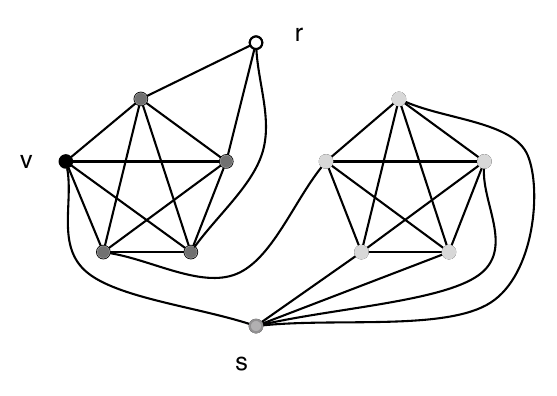}
    \label{fig:gadgets2}
  }
  \subfigure[$k=5$, $\ell =2$ and $v= s$]{
    \includegraphics[scale=.57]{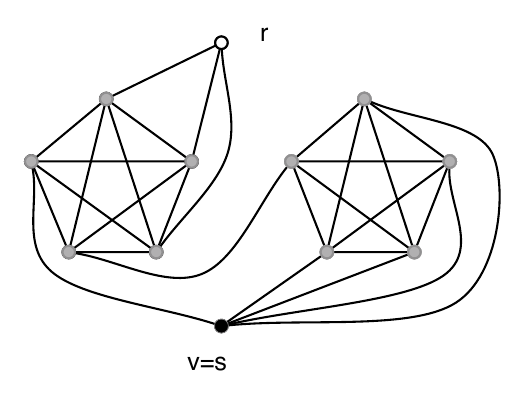}
    \label{fig:gadgets3}
  }

  \caption[Optional caption for list of figures]{Above are examples
    for the gadgets when $k=5$ and in the case of Figure~\ref{fig:gadgets1} $\ell
    = 3$ or in Figures~\ref{fig:gadgets2}~and~\ref{fig:gadgets3} $\ell
    =2$. The shading denotes the stripping order as used in
    Lemma~\ref{set_gadget}, we start by stripping the darkest vertex
    and lastly strip the lightest vertex.}
  \label{fig:subfigureExample}
  \end{figure}

\begin{description}
\item[If $\ell$ is odd] $H_{\ell}$ is composed of 2 copies of $K_k$, a
  clique on $k$ vertices, and a root vertex $r$. Connect $r$ to
  $(k-\ell)/2$ of the vertices in each of the complete graphs. Add a
  matching between the vertices of degree $k-1$ of the $K_k$
  subgraphs. (Figure~\ref{fig:gadgets1})

\item[If $\ell$ is even] $H_{\ell}$ is composed of a left and right copy of $K_k$,
  a root vertex $r$, and an extra vertex $s$. Connect $r$ to $k-\ell$
  vertices in the left $K_k$. Connect $s$ to $\ell/2$ vertices
  in the left $K_k$. Connect $s$ to $k-\ell/2$ of the vertices in the
  right $K_k$. Add a matching between the vertices of degree $k-1$ of the left and right $K_k$
  subgraphs. (Figures~\ref{fig:gadgets2} and ~\ref{fig:gadgets3})
\end{description}

It is easy to verify that $H_{\ell}$ has the following properties:

\begin{enumerate}
\item All the vertices, except the root vertex $r$, have degree $k$.
\item The root vertex $r$ has degree $k-\ell$.
\item $H_\ell \backslash \{r\}$ is connected.
\item $H_\ell$ is $(k-1)$-strippable.
\end{enumerate}

Let $\mathcal{S} =\{ S_1, S_2, \dots S_m \}$ be the instance of set
cover. We wish to find $\ell$ sets that cover all of the elements. Let
$f_x$ denote the frequency of element $x$ in $\mathcal S$. Let $k$ be
the smallest odd number which is greater than $\textrm{max}_{i,x} \{|S_i|,f_x\}$.

 \begin{figure}[ht]
    \centering
    \includegraphics[scale=.4]{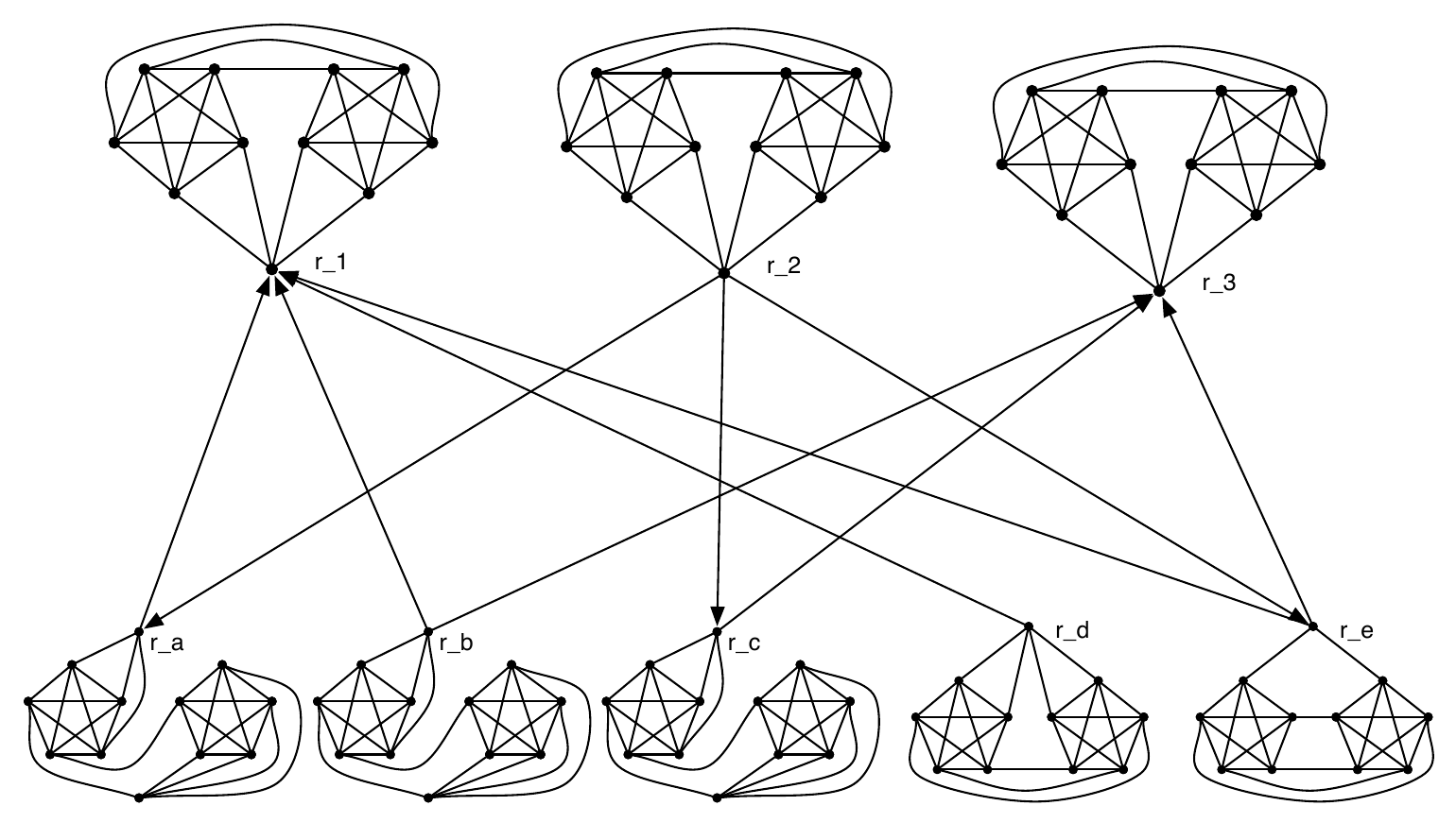}

  \caption[]{The following is an example of the graph corresponding to the Set
Cover instance: $\{a,b,c,d,e\}$ with sets $\{S_1,S_2,S_3\}$,
$S_1=\{a,b,d,e\}$, 
$S_2=\{a,c,e\}$, $S_3=\{b,c,e\}$. Suppose that $S_1$ and $S_3$ are chosen for the
cover, then we see all edges from the element gadgets are directed
toward $S_1$ and $S_3$.}
  \label{fig:gadget_graph_clarified}
  \end{figure}

We construct a graph $G$ as follows (See Figure~\ref{fig:gadget_graph_clarified}):

For each $S_i \in \mathcal{S}$ create a {\em set} $k$-gadget $H_1$, with root vertex
$r_i$. For each element $x$ create an {\em element} $k$-gadget $H_{f_x}$, with root vertex
$r_x$. For every $x \in S_i$ connect $r_i$ to 
$r_x$. $G$ has the following properties:

\begin{enumerate}
\item All vertices in $G$ have degree $k$ except
for the vertices $r_i$ of the set gadgets, these have degree $k+|S_i|-1$.
\item All vertices have degree at
least $k$, so the minimum possible maximum indegree of any acyclic
orientation of $G$ is $k$.
\item $G$ is $k$-strippable.
\end{enumerate}

The first two properties are clear from construction. We will prove the third property with the following two lemmata.

  \begin{lemma}
    \label{set_gadget}
    The vertices of a set gadget in $G$ with any vertex $v$ removed are
    $(k-1)$-strippable.
  \end{lemma}

  \begin{proof}
    $r$ is the only vertex in the set gadget that potentially has
    degree greater than $k$. 

    If $v=r$ then the vertices that were adjacent to $r$ have degree
    $k-1$ so we can $(k-1)$-strip the remaining
    vertices.

    Suppose that $v\neq r$. The vertices of the set gadget can be
    $(k-1)$-stripped as illustrated in Figure~\ref{fig:subfigureExample}. 

If $\ell$ is odd, first strip all of the
    vertices of the clique containing $v$ then strip all of the
    vertices of the other clique. We know that $r$ originally had
    degree $k+|S_i|-1$ and we removed $k-1$ of the vertices adjacent
    to it. Thus $r$ has degree $|S_i|$ which is at most $k-1$, so we
    can strip $r$.

    If $\ell$ is even, then either $v\in K_k$ or $v=s$. If $v\in K_k$ for
    either the left or right clique, first strip all of the vertices
    in this clique, then strip $s$, then strip the vertices in the
    other clique and finally strip $r$ as in the odd case. If $v=s$ then we can strip
    both cliques, because both were connected to $s$, and finally
    strip $r$.
  \end{proof}

A similar argument shows:

\begin{lemma}
  \label{element_gadget}
     Any element gadget with any vertex $v$ removed is
    $(k-1)$-strippable.
\end{lemma}

It follows that $G$ is $k$-strippable: stripping one vertex of degree $k$ from each
set and element gadget leaves $(k-1)$-strippable subgraphs.

\begin{theorem}
$G$ has an acyclic orientation with at most $\ell$ vertices of
indegree $k$ if and only if there is a covering subcollection of
$\mathcal S$ of size at most $\ell$.
\end{theorem}

\begin{proof}
  For the forward direction: Let $X$ be the set of at
    most $\ell$ vertices which have indegree $k$. $G\backslash X$ is
  $(k-1)$-strippable. Let $\mathcal S^\prime$ be the subcollection of
  $\mathcal S$ containing

  (i) all the sets whose gadgets have an indegree $k$ vertex and

  (ii) for each element gadget that has an indegree $k$ vertex, one set
  that contains this element.

  Notice that $|\mathcal S^\prime| \leq \ell$. We will show that
  $\mathcal S^\prime$ is a covering.

  For any element $x$, let $P$ be the element gadget corresponding to
  $x$. There are two cases for $P$:

  (a) $P$ includes a vertex of indegree $k$.

  In this case $x$ is covered by a set of type (ii). 

 (b) $P$ does not include a vertex of indegree $k$.

 The degree of the root vertex $r_P$ of $P$ is $k$. The
   assumption that $P$ does not have any vertices of indegree $k$
 implies that at least one edge is oriented away from $r_P$. If this
 edge is from $r_P$ to a set gadget, then at least one vertex of the
 set gadget must have indegree $k$. The set corresponding to this set
 gadget must be included in $\mathcal S^\prime$, so $x$ is covered. If
 the edge oriented away from $r_P$ is directed to another vertex in
 $P$, then there must be a vertex of indegree $k$ in $P$. By the
 construction of the gadget and the acyclicity property of the
 orientation, this is a contradiction to the fact that the edge
 oriented away from $r_P$ is directed into another vertex in $P$. Thus
 in any case, element $x$ is covered.

Therefore $\mathcal S^\prime$ is a cover.

For the reverse direction: Let $\mathcal S^\prime$ be the collection
of at most $\ell$ sets from $\mathcal S$ that form a cover. Take a
non-root vertex from each set gadget corresponding to a set in
$\mathcal S^\prime$ and orient all edges toward it. Each of these
gadgets is now $(k-1)$-strippable by Lemma~\ref{set_gadget}. Each
element $x$ is covered, so stripping the set gadget covering $x$
directs the edge between the set gadget and the element gadget for $x$
away from the element gadget (see the orientation of
Figure~\ref{fig:gadget_graph_clarified}). The root of the element
gadget has degree $k-1$, so each element gadget is $(k-1)$-strippable
in $G\backslash X$. Consider the set gadgets for the sets not in
$\mathcal S^\prime$. The roots of these gadgets all have remaining
degree $k-1$, because all of the element gadgets have been stripped so
these are also $k-1$-strippable.  This orientation has at most $\ell$
vertices of indegree $k$.
\end{proof}

\section{Closing}

Graph orientation is a rich problem area.  In this paper we have
presented three variants of the problem with their respective
motivations.  In one variant the resulting graph has no structural
constraints, in another strong connectivity is required, and finally
an acyclic orientation is required.  For the first two variants the
simple algorithm of path reversal proves to be powerful.  We have
shown the optimality of the algorithm for the ``minimizing the maximum
indegree'' objective in both variants and for the minimizing the
lexicographic order objective for
the first variant.  We conjecture that SC-path-reversal is indeed
optimal for the ``minimizing the lexicographic order'' objective as well. The third variant,
requiring the resulting graph to be acyclic introduced quite a
different problem.  We have included an NP-hardness proof for acyclic
minimizing the lexicographic order to demonstrate the point.  How to
approximate the ``minimizing the lexicographic order''
objective and enforce the acyclicity constraint remains an interesting
open problem.

\section{Acknowledgements}

We thank the reviewers for their help with making our proofs more
clear. We would also like to thank our reviewers for pointing us to an
alternate proof for Theorem 6~\cite{fg1994}.

\bibliographystyle{plain}
\bibliography{indegree}

\end{document}